\documentclass[a4paper,UKenglish,cleveref, autoref, thm-restate]{lipics-v2021}

\pdfoutput=1 
\hideLIPIcs  

\usepackage{complexity}
\usepackage{mathtools}
\usepackage{bm}
\usepackage{amsthm}
\usepackage{amssymb}
\usepackage{xspace}
\usepackage{hyperref}   
\usepackage{cleveref}
\usepackage{enumitem}
\usepackage{booktabs}
\usepackage{tikz}
\usepackage[table]{xcolor}
\usetikzlibrary{positioning}

\newcommand{\mdfull}{\textsc{Metric Dimension}}
\newcommand{\gsfull}{\textsc{Geodetic Set}}
\newcommand{\smdfull}{\textsc{Strong Metric Dimension}}

\newcommand{\calO}{\mathcal{O}}

\DeclareMathOperator{\dist}{\mathtt{dist}}

\newcommand{\smd}{\mathtt{smd}}
\newcommand{\vc}{\mathtt{vc}}
\newcommand{\fvs}{\mathtt{fvs}}

\newcommand{\tw}{\mathtt{tw}}
\newcommand{\pw}{\mathtt{pw}}
\newcommand{\diam}{\mathtt{diam}}

\newcommand{\ETH}{\textsf{ETH}}

\newcommand{\g}{\mathtt{g}}
\newcommand{\f}{\mathtt{f}}

\newcommand{\defproblem}[3]{
  \vspace{1mm}
\noindent\fbox{
  \begin{minipage}{0.96\textwidth}
  \begin{tabular*}{\textwidth}{@{\extracolsep{\fill}}lr} #1 \\ \end{tabular*}
  {\bf{Input:}} #2  \\
  {\bf{Question:}} #3
  \end{minipage}
  }
  \vspace{1mm}
}

\title{On the Hardness of Strong Metric Dimension} 

\author{Prafullkumar Tale}{
Indian Institute of Science Education and Research Pune, India \and \url{https://pptale.github.io/}}{prafullkumar@iiserpune.ac.in}{https://orcid.org/0000-0001-9753-0523}{Supported by the ARG-MATRICS Grant MaPLE (ARGM/2026/0089).}

\authorrunning{P. Tale} 

\Copyright{Prafullkumar Tale} 



\category{} 

\relatedversion{} 

\nolinenumbers 

\EventEditors{John Q. Open and Joan R. Access}
\EventNoEds{2}
\EventLongTitle{42nd Conference on Very Important Topics (CVIT 2016)}
\EventShortTitle{CVIT 2016}
\EventAcronym{CVIT}
\EventYear{2016}
\EventDate{December 24--27, 2016}
\EventLocation{Little Whinging, United Kingdom}
\EventLogo{}
\SeriesVolume{42}
\ArticleNo{23}

\begin{document}

\maketitle

\begin{abstract}
Let \(G\) be a connected simple undirected graph.
A vertex \(w\) is said to \emph{strongly resolve} a pair of distinct vertices
\(u, v \in V(G)\) if either there exists an isometric path (i.e.~a shortest path) from 
\(w\) to \(u\)
that contains \(v\), or there exists an isometric path from \(w\) to  \(v\)
that contains \(u\).
A subset \(S \subseteq V(G)\) is said to \emph{strongly resolve} 
\(G\)
if every pair of distinct vertices of \(G\) is strongly resolved
by at least one vertex in \(S\).
In the \textsc{Strong Metric Dimension} problem, the input consists of a graph
\(G\) and a positive integer \(k\), and the objective is to determine whether
there exists a subset \(S \subseteq V(G)\) of size at most \(k\) that strongly
resolves \(G\).
In this article, we show that \textsc{Strong Metric Dimension} is
\NP-complete even on \((i)\) graphs of diameter two, and 
\((ii)\) graphs of constant pathwidth and constant feedback vertex set number.
\end{abstract}

\section{Introduction}

Consider a connected simple undirected graph \(G\), 
and suppose \(\dist(x,y)\) denotes the length of an isometric (i.e.~shortest) path between \(x\) and \(y\).
A vertex \(v \in V(G)\) is said to \emph{distinguish} 
two vertices \(x\) and \(y\) if \(\dist(v,x) \neq \dist(v,y)\).
A set \(S \subseteq V(G)\) is called a \emph{metric generator} or \emph{resolving set} for \(G\)
if every pair of vertices of \(G\) is distinguished by some vertex in \(S\).
The cardinality of a minimum metric generator is called the
\emph{metric dimension} of \(G\).
Although vertices in a metric generator distinguish all
pairs of vertices in a graph, they do not uniquely determine all distances
in the graph.
For example, the two graphs in Figure~\ref{fig:intro-fig} have identical distance vectors with respect to the metric generator
\(\{a,b\}\), 
yet the actual distances between vertices that share identical distance representations differ between the two graphs.
Seb\H{o} and Tannier~\cite{sebo04} observed this, and proved that if the notion of a `metric generator' is replaced by the stronger notion of a
\emph{strong metric generator}, then all distances in the graph can be
uniquely determined.
For more discussion, we refer the reader to the
PhD thesis of Kuziak~\cite{kuziak2014strong}.

A vertex \(w \in V(G)\) is said to \emph{strongly resolve} two distinct vertices
\(u, v \in V(G)\) if \(\dist(w,u) = \dist(w,v) + \dist(v,u)\) or \(\dist(w,v) = \dist(w,u) + \dist(u,v)\),
that is, if there exists an isometric path from \(w\) to \(u\) containing \(v\),
or an isometric path from \(w\) to \(v\) containing \(u\).
A set \(S \subseteq V(G)\) is a \emph{strong metric generator} for a connected
graph \(G\) if every pair of vertices of \(G\) is strongly resolved by some
vertex in \(S\).
The minimum cardinality of a strong metric generator for \(G\) is called the
\emph{strong metric dimension} of \(G\) and is denoted by \(\smd(G)\).
In the \textsc{Strong Metric Dimension} problem, the input is a graph
\(G\) and an integer \(k\), and the objective is to decide whether the strong metric dimension of \(G\) is at most \(k\).

Recently, there has been renewed interest in the complexity of \textsc{Strong Metric Dimension}.
Foucaud et al.~\cite{DBLP:conf/icalp/FoucaudGK0IST24} showed that it is one
of the three \NP-complete problems admitting double-exponential
dependence on the treewidth \((\tw)\) of the input graph.
The other two problems are \textsc{Metric Dimension} and \textsc{Geodetic Set}.
For these two problems, the input is a graph \(G\) and a positive integer
\(k\).
In \textsc{Metric Dimension}, the task is to determine whether
the metric dimension of \(G\) is at most \(k\), whereas in the \textsc{Geodetic Set}\ problem, it is to determine whether there exists a set
\(S \subseteq V(G)\) of size at most \(k\) such that, for every vertex \(u\) in \( V(G)\), there exist two vertices \(s_1, s_2\)
in \(S\) for which an isometric path between \(s_1\) and \(s_2\) contains \(u\).
Foucaud et al.~\cite{DBLP:conf/icalp/FoucaudGK0IST24} proved that, unless the
Exponential Time Hypothesis (\ETH) fails, neither \textsc{Metric Dimension} nor \textsc{Geodetic Set}\ admits
an algorithm running in time
\(2^{2^{o(\tw)}} \cdot n^{\mathcal{O}(1)}\), even on graphs with constant diameters (\(\diam\)),
where \(n\) is the number of vertices.
For \textsc{Strong Metric Dimension}, they obtained even stronger lower bounds, ruling out algorithms
running in time
\(2^{2^{o(\vc)}} \cdot n^{\mathcal{O}(1)}\) under the same assumption, where
\(\vc\) denotes the vertex cover number of the input graph.
They also showed that these bounds are tight by presenting algorithms running
in time
\(2^{\diam^{\mathcal{O}(\tw)}} \cdot n^{\mathcal{O}(1)}\) for \textsc{Metric Dimension} and
\textsc{Geodetic Set}, and
\(2^{2^{\mathcal{O}(\vc)}} \cdot n^{\mathcal{O}(1)}\) for \textsc{Strong Metric Dimension}.

Foucaud et al.~\cite{DBLP:conf/icalp/FoucaudGK0IST24}
motivated the use of larger parameters, such as treewidth plus diameter,
by highlighting the complexity of these problems under smaller
parameterizations.
The problem \textsc{Metric Dimension} is \NP-complete even on graphs of diameter two~\cite{FoucaudMNPV17b},
as well as on graphs of pathwidth \(24\)~\cite{LM21}.
Similarly, \textsc{Geodetic Set}\ is \NP-complete even on graphs of diameter two~\cite{floCALDAM20}.
Recently, it was shown that this problem remains \NP-complete even on graphs
of constant pathwidth and constant feedback vertex set number~\cite{DBLP:conf/iwpec/Tale25}.
However, for \textsc{Strong Metric Dimension}, hardness results parameterized by diameter or treewidth
are not stated explicitly in~\cite{DBLP:conf/icalp/FoucaudGK0IST24}, and to the
best of our knowledge, such results have not appeared in the literature.
See Table~\ref{table:overview} for an overview of the results mentioned above.

\begin{figure}[t]
\centering
\includegraphics[scale=0.6]{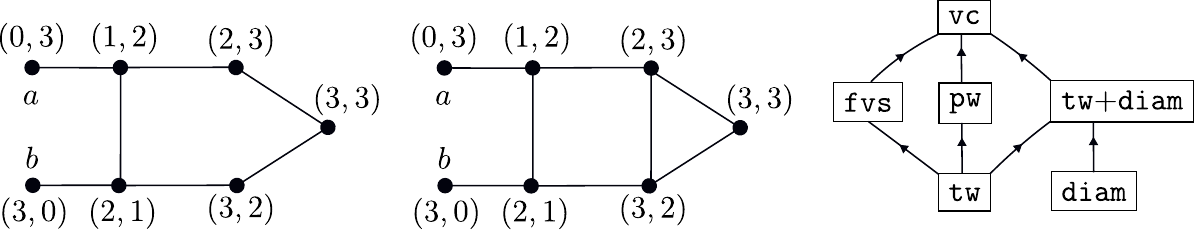}
\caption{(Left and center) Two different graphs with the same metric generators and respective distances. 
(Right) Hierarchy of the parameters mentioned in this article:
vertex cover number \((\vc)\),
feedback vertex set \((\fvs)\),
pathwidth \((\pw)\), treewidth \((\tw)\), and diameter \((\diam)\).\label{fig:intro-fig}}
\vspace{-0.5cm}
\end{figure}

In this article, we address this gap in the literature on \textsc{Strong Metric Dimension}.
A closer inspection of the reduction presented in~\cite{DBLP:conf/icalp/FoucaudGK0IST24}
suggests that \textsc{Strong Metric Dimension}\ is \NP-complete even on graphs of constant diameter.
However, due to the technical nature of the reduction, it is difficult to
determine the exact bound on the diameter.
An examination of the gadgets involved, and their connections indicate that the resulting graphs have diameter
strictly larger than three.
As our first result, we present a significantly simpler reduction that achieves
the best possible bound on the diameter 
when the problem is \NP-complete.

\begin{theorem}\label{thm:smd-np-hard-diam}
\textsc{Strong Metric Dimension} remains \NP-complete even on graphs of diameter two.
\end{theorem}

Next, we investigate whether \textsc{Strong Metric Dimension}\ becomes fixed-parameter tractable when
parameterized by measures smaller than the vertex cover number, such as
pathwidth or the feedback vertex set number.
We resolve this question negatively.

\begin{theorem}\label{thm:np-hardness-fvs}
\textsc{Strong Metric Dimension}\ remains \NP-complete even on graphs of the feedback vertex set number
\(25\) and pathwidth \(27\).
\end{theorem}

The next natural question is whether \textsc{Strong Metric Dimension}\ is fixed-parameter tractable when parameterized by treewidth plus diameter.
A simple application of Courcelle's 
theorem~\cite{Courcelle90} implies that this is indeed the case.
Although Monadic Second-Order (MSO$_2$) logic cannot express unbounded shortest-path distances,
for any fixed distance $d \leq \diam(G)$, the property that the exact shortest path 
distance between two vertices $u$ and $v$ is $d$ can be expressed in first-order logic by existentially quantifying over paths of length up to $d$. 
Consequently, for a set \(S \subseteq V(G)\), the condition that a vertex $w \in S$ strongly resolves the pair $(u, v)$, which requires
encoding either 
$\dist(w, u) = \dist(w, v) + \dist(v, u)$ or 
$\dist(w, v) = \dist(w, u) + \dist(u, v)$, 
can be written as an MSO$_2$ formula 
$\Phi(S)$, whose length depends only on $\diam$.
By Courcelle's Theorem, the minimum cardinality of a set $S$ satisfying $\Phi(S)$ can be computed in time $O(f(|\Phi|, \tw(G)) \cdot n)$. 
This yields an FPT algorithm parameterized by $\tw(G) + \diam(G)$.
An interesting question is to determine whether \smdfull\ admits an
algorithm running in time $2^{\diam^{\mathcal{O}(\tw)}}\cdot n^{\calO(1)}$, similar to the known
algorithms for \mdfull\ and \gsfull\
in~\cite{DBLP:conf/icalp/FoucaudGK0IST24}.

\begin{table}[t]
    \centering
    \renewcommand{\arraystretch}{1.5}
    \begin{tabular}{p{0.13\linewidth}|p{0.16\linewidth}|p{0.15\linewidth}|p{0.1\linewidth}|p{0.1\linewidth}|p{0.12\linewidth}|}
        \toprule
        & \(\tw\) + \(\diam\) & \(\vc\) & \(\diam\) &  \(\fvs+\pw\) & \(\tw\)   \\
        \midrule
        \textsc{Metric} \newline \textsc{Dimension} & \(2^{\diam^{\calO(\tw)}}\) \newline
        No \(2^{\diam^{o(\tw)}}\)~\cite{DBLP:conf/icalp/FoucaudGK0IST24}  & \(2^{\calO(\vc^2)}\) \newline 
        No \(2^{o(\vc^2)}\)~\cite{DBLP:conf/stacs/FoucaudGK0IST25} &
        \textsf{paraNP}-hard~\cite{FoucaudMNPV17b}
         & \W[1]-hard~\cite{GKIST23} \newline  &  \textsf{paraNP}-hard~\cite{LM21}
        \\
        \midrule
        \textsc{Geodetic Set} & \(2^{\diam^{\calO(\tw)}}\) \newline
        No \(2^{\diam^{o(\tw)}}\)~\cite{DBLP:conf/icalp/FoucaudGK0IST24} & \(2^{\calO(\vc^2)}\) \newline 
        No \(2^{o(\vc^2)}\)~\cite{DBLP:conf/stacs/FoucaudGK0IST25} & \textsf{paraNP}-hard~\cite{floCALDAM20} & \multicolumn{2}{c|}{\textsf{paraNP}-hard~\cite{DBLP:conf/iwpec/Tale25}} \\
        \midrule
        \textsc{Strong Metric Dimension} & \FPT & \(2^{2^{\calO(\vc)}}\) \newline
        No \(2^{2^{o(\vc)}}\)~\cite{DBLP:conf/icalp/FoucaudGK0IST24} & \textsf{paraNP}-hard (Thm~\ref{thm:smd-np-hard-diam}) & \multicolumn{2}{c|}{\textsf{paraNP}-hard (Thm~\ref{thm:np-hardness-fvs})} \\
        \bottomrule
    \end{tabular}
    \caption{Structural parameterized complexity of the three problems mentioned in the introduction.
    We omit the polynomial factors in the running time.
    All the lower bounds are based on~\textsf{ETH}.\label{table:overview}}
    \vspace{-0.75cm}
\end{table}

\subparagraph*{Related Work}
The computational complexity of determining 
the strong metric dimension of a graph has been studied in the literature.
In the introductory paper, Seb\H{o} and Tannier~\cite{sebo04} used
the \textsc{Strong Metric Dimension}
problem to design an efficient algorithm for \textsc{Connected Join Existence}.
Oellermann and Peters-Fransen~\cite{DBLP:journals/dam/OellermannP07}
showed an instance \((G,k)\) of \textsc{Strong Metric Dimension}\ 
can be reduced in polynomial time to an instance \((G_{SR},k)\) of
\textsc{Vertex Cover}, where \(G_{SR}\) is called a \emph{strong resolving graph}
of \(G\) (see Definition~\ref{def:strong-resolving-graph}).
Consequently, many algorithmic and hardness results known for
\textsc{Vertex Cover} carry over to \textsc{Strong Metric Dimension}~\cite{DBLP:journals/dam/KuziakPRY18}.
In particular, \textsc{Strong Metric Dimension}\ is \NP-complete but fixed-parameter tractable when
parameterized by the solution size.
More involved conditional lower bounds and hardness results for
\textsc{Vertex Cover} also transfer to \textsc{Strong Metric Dimension}~\cite{DBLP:journals/dam/DasGuptaM17}.

We remark that results on structural parameterizations for
\textsc{Vertex Cover} on the strong resolving graph \(G_{SR}\) of \(G\) do not directly
translate to results for \textsc{Strong Metric Dimension}\ on the original graph \(G\).
This is because the strong resolving graph \(G_{SR}\) of a graph \(G\) does not
necessarily preserve the structural properties of \(G\).
For instance, if \(G\) is a star with \(n\) leaves, then \(G_{SR}\) is a complete
graph on \(n\) vertices together with an isolated vertex.
In this case, the treewidth of \(G_{SR}\) is significantly larger than that of
\(G\).
Therefore, even though \textsc{Vertex Cover} is fixed-parameter tractable when
parameterized by treewidth, this does not yield an \FPT\ algorithm
for \textsc{Strong Metric Dimension} parameterized by treewidth.
Consequently, the methods used in~\cite{DBLP:journals/dam/DasGuptaM17}
to translate conditional lower bounds
are not applicable when we consider structural parameters.

\section{Preliminaries}
\label{sec:appendix-prelims}

We use standard graph-theoretic notation, and we refer the 
reader to~\cite{D12} for any undefined notation. 
For an undirected graph $G$, sets $V(G)$ and $E(G)$ denote its set of vertices and edges, respectively.
Two vertices $u,v\in V(G)$ are {\it adjacent} or {\it neighbors} if $uv\in E(G)$. 
The {\it open neighborhood} of a vertex $u\in V(G)$, denoted by $N(u):=N_G(u)$, is the set of vertices that are neighbors of $u$. 
The {\it closed neighborhood} of a vertex $u\in V(G)$ is denoted by $N[u]:=N_G[u]:=N_G(u)\cup \{u\}$.
For any $u,v \in V(G)$, we say that $u$ is connected to $v$ by a path $P$ of length $\ell$ if $P= w_0w_1\ldots w_\ell$, where $w_0 = u$ and $v = w_\ell$.
The {\it distance} between two vertices $u,v\in V(G)$ in $G$, denoted by $\dist(u,v):=\dist_G(u,v)$, is the length of a 
$(u,v)$-shortest path in $G$. 
For a subset $S$ of $V(G)$, we define $N[S] = \bigcup_{v \in S} N[v]$ and $N(S) = N[S] \setminus S$.
For a subset $S$ of $V(G)$, we denote the graph obtained by deleting $S$ from $G$ by $G - S$.
The {\it complement} of a graph $G$ is a graph $H$ with the same vertex set, and such that any two vertices $u,v\in V(G)$ are adjacent in $H$ if and only if they are not adjacent in $G$.


\begin{definition}
\label{def:mutually-max-distant}
Given a graph $G$, we say a vertex $u\in V(G)$ is maximally distant from $v\in V(G)$ if
there is no $x \in V(G) \setminus \{u\}$ such that a shortest path between $x$ and $v$ contains $u$.
Formally, for every 
$y \in N(u)$, we have
$\dist(y, v) \le \dist(u, v)$.
\end{definition}
Consider an auxiliary graph $G_{SR}$
of $G$ defined as follows.
\begin{definition}
\label{def:strong-resolving-graph}
Given a connected graph $G$, the  \emph{strong resolving graph} of $G$,  denoted by $G_{SR}$, has vertex set 
$V(G)$ and two vertices $u, v$ are adjacent if and only if
$u$ and $v$ are mutually maximally  distant in $G$.
\end{definition}
For any two mutually maximally distant
vertices in $G$, there is no vertex in $G$ 
that strongly resolves them, except themselves. 
Hence, if $u$ and $v$ are mutually maximally  distant in $G$, 
then, for any strong resolving set $S$ of 
$G$, at least one of $u$ or $v$ is in $S$,
i.e., $|\{u, v\}\cap S|\geq 1$.
Oellermann and Peters-Fransen~\cite[Theorem 2.1]{DBLP:journals/dam/OellermannP07} showed that
this necessary condition is also sufficient.

\begin{proposition}[\cite{DBLP:journals/dam/OellermannP07}]
\label{prop:smd-vc}
For any connected graph \(G\), we have \(\smd(G) = \vc(G_{SR})\).
\label{prop:smd-to-vc}
\end{proposition}

A set of vertices $Y$ is said to be {an} \emph{independent set} if no two vertices in $Y$ are adjacent.
For a graph $G$, a set $X \subseteq V(G)$ is said to be {a} \emph{vertex cover} if $V(G) \setminus X$ is an independent set.
A vertex cover $X$ is a \emph{minimum vertex cover} if for any other vertex cover $Y$ of $G$, we have $|X| \le |Y|$.
\emph{The vertex cover number} of graph $G$ is the size of {a} minimum vertex cover of {a graph} $G$.
For a graph $G$, a set $X \subseteq V(G)$ is said to be {a} \emph{feedback vertex set} if $V(G) \setminus X$ is an acyclic graph.
We define the notation of \emph{the feedback vertex set number} in the analogous way.
A \emph{path decomposition} of $G$ is a sequence of subsets
$\mathcal{P} = (X_1, X_2, \dots, X_r)$ of $V$ such that:
\((i)\) $\displaystyle \bigcup_{i=1}^{r} X_i = V$.
\((ii)\) For every edge $uv \in E$, there exists an index $i \in [r]$ such that  $\{u,v\} \subseteq X_i$.
\((iii)\) For every vertex $v \in V$, the set $\{\, i \mid v \in X_i \,\}$ forms a contiguous interval of $\{1,\dots,r\}$.
The sets $X_1,\dots,X_r$ are called \emph{bags}.
The \emph{width} of a path decomposition 
$\mathcal{P} = (X_1,\dots,X_r)$ is \(\max_{1 \le i \le r} |X_i| - 1\).
The \emph{pathwidth} of a graph $G$, denoted $\pw(G)$, is
the minimum width over all path decompositions of $G$.
For details on parameterized complexity and related terminologies, we refer the reader to the recent book by Cygan et al.~\cite{DBLP:books/sp/CyganFKLMPPS15}.

\section{NP-hardness on Graphs of Constant Diameter}\label{sec:np-hardness-diam}

In this section, we prove that \textsc{Strong Metric Dimension}\ is \NP-complete even when restricted to graphs of diameter \(2\).
We begin with the standard polynomial-time reduction that, given an instance \(\varphi\) of \textsc{3-SAT}, 
constructs an equivalent instance \((H,k)\) of \textsc{Vertex Cover}.
Next, we construct a graph \(G\) by taking the complement graph \(\overline{H}\) 
and adding a new global vertex \(g\) that is adjacent to every vertex of \(\overline{H}\).
We then show that the strong resolving graph \(G_{SR}\) of \(G\) is isomorphic to \(H\), 
up to the presence of the additional vertex \(g\), which does not affect the vertex 
cover number of \(G_{SR}\) as \(g\) is an isolated vertex in \(G_{SR}\).
Observe that \(G\) is connected and has diameter \(2\), since every pair of 
non-adjacent vertices in \(G\) has a common neighbor \(g\).
Finally, by Proposition~\ref{prop:smd-to-vc}, we have \(\smd(G) = \vc(G_{SR})\).
Since \(G_{SR}\) is isomorphic to \(H\) (modulo the vertex \(g\)), it follows that \(\smd(G) = \vc(H)\).
Consequently, the instance \((H,k)\) of \textsc{Vertex Cover} is a yes-instance if 
and only if \((G,k)\) is a yes-instance of \textsc{Strong Metric Dimension}, establishing the correctness of the reduction.

Recall the reduction that takes as input 
an instance \(\varphi\) of \textsc{3-SAT} with \(n\) variables and \(m\) clauses, and constructs graph \(H\) as follows:
\((i)\) For each variable \(x_i\), it creates two vertices labeled \(x_i\) and
\(\neg {x_i}\), and add an edge between them.
\((ii)\) For each clause \(C_j = \langle \ell_{j,1} \lor \ell_{j,2} \lor \ell_{j,3} \rangle\),
creates a triangle with vertices labeled \(\ell_{j,1}, \ell_{j,2}, \ell_{j,3}\).
\((iii)\) For each occurrence of a literal \(\ell\) in a clause \(C_j\),
add an edge between the clause vertex labeled \(\ell\) and 
the corresponding variable vertex \(x_i\) or \(\neg{x_i}\).
This completes the construction of \(H\).
The reduction returns \((H, k)\), where \(k = n + 2m\), as an instance of 
\textsc{Vertex Cover}.

The correctness of the reduction is proved in~\cite[Section~3.1.3]{DBLP:books/fm/GareyJ79}.
At this stage, we highlight that \(H\) has no isolated vertex or false twins.
The following lemma is useful to construct an instance \((G, k)\) of \textsc{Strong Metric Dimension}.

\begin{lemma}\label{lemma:compliment-strong-res-graph}
Let \(H\) be a graph on at least three vertices 
with no isolated vertices and no false twins.
Let \(G\) be the graph obtained by adding a global vertex \(g\) to the
complement graph \(\overline{H}\).
Then, \(H\) and \(G_{SR} - \{g\}\) are identical graphs.
\end{lemma}
\begin{proof}
By the definition, it is easy to see that \(V(H) = V(G_{SR} \setminus \{g\})\).
It remains to prove that \(E(H) = E(G_{SR})\).

\((\Rightarrow)\)
Let \(uv \in E(H)\).
By definition, the vertices \(u\) and \(v\) are non-adjacent in \(\overline{H}\),
and consequently also non-adjacent in \(G\).
Since \(G\) contains a global vertex \(g\), the distance between any two
vertices in \(G\) is at most two.
In particular, \(d_G(u,v)=2\).
For every vertex \(y \in N_G(u)\), we have
\(d_G(y,v) \le 2 = d_G(u,v)\).
By symmetry, for every vertex \(x \in N_G(v)\), we have
\(d_G(x,u) \le 2 = d_G(u,v)\).
Thus, \(u\) and \(v\) are mutually maximally distant in \(G\).
By Definition~\ref{def:strong-resolving-graph}, this implies that
\(uv \in E(G_{SR})\).
Since \(uv\) was an arbitrary edge of \(H\), we obtain
\(E(H) \subseteq E(G_{SR})\).

\((\Leftarrow)\)
Let \(uv \in E(G_{SR})\).
Since \(g\) is an isolated vertex in \(G_{SR}\), it follows that \(u \neq g\) and \(v \neq g\).
By the definition of the strong resolving graph, the vertices \(u\) and \(v\) are mutually maximally distant in \(G\).
We distinguish two cases, depending on whether \(u\) and \(v\) are adjacent in \(G\) or not.
In the first case, since \(u\) and \(v\) are mutually maximally distant, for every \(y \in N_G(u)\) we have
\(\dist_G(y,v) \le \dist_G(u,v) = 1\).
Hence, every neighbor \(y\) of \(u\) is adjacent to \(v\), which implies \(N_G(u) \subseteq N_G[v]\).
By symmetry, we also obtain \(N_G(v) \subseteq N_G[u]\). Consequently, \(N_G[u] = N_G[v]\).
Recall that \(G\) is obtained from \(\overline{H}\) by adding the global vertex \(g\).
Since \(u \neq g\) and \(v \neq g\), the equality \(N_G[u] = N_G[v]\) implies that
\(u\) and \(v\) have identical neighborhoods in \(\overline{H}\), and therefore also in \(H\).
Thus, \(u\) and \(v\) are false twins in \(H\), contradicting the assumption that
\(H\) contains no false twins. Hence, this case cannot occur.
In the second case, since \(G\) contains a global vertex \(g\), the distance between 
any two distinct non-adjacent vertices in \(G\) is exactly \(2\).
Since \(G = \overline{H}\) together with the global vertex \(g\), the non-adjacency of 
\(u\) and \(v\) in \(G\) implies that they are adjacent in \(H\).
Thus, \(uv \in E(H)\).
Since \(uv\) was an arbitrary edge of \(G_{SR}\), we conclude that
\(E(G_{SR}) \subseteq E(H)\).
This completes the proof of the lemma.
\end{proof}

\begin{proof}[Proof of Theorem~\ref{thm:smd-np-hard-diam}]
It is straightforward to verify that the problem belongs to \textsc{NP}.
We describe a polynomial-time reduction from \textsc{3-SAT}.
Given an instance \(\varphi\) of \textsc{3-SAT}, it first constructs an instance
\((H,k)\) of \textsc{Vertex Cover}.
Next, it constructs a graph \(G\) by adding a single global vertex \(g\) to the
complement graph \(\overline{H}\).
The resulting instance \((G,k)\) is returned as an instance of \textsc{Strong Metric Dimension}.

The correctness of the first step of the reduction follows from the standard
reduction from \textsc{3-SAT} to \textsc{Vertex Cover}
(see~\cite[Section~3.1.3]{DBLP:books/fm/GareyJ79}).
The correctness of the second step follows from
Lemma~\ref{lemma:compliment-strong-res-graph}, which implies that
\(H\) and \(G_{SR} - \{g\}\) are identical graphs, together with 
Proposition~\ref{prop:smd-to-vc}, we have \(\smd(G) = \vc(G_{SR}) = \vc(H)\).
\end{proof}
\section{NP-hardness on Graphs of Constant Pathwidth and Feedback Vertex Set Number}\label{sec:np-hardness-fvs}

In this section, we prove Theorem~\ref{thm:np-hardness-fvs}.  
To this end, we present a reduction from a variation of \textsc{3-SAT},
called \textsc{3-Partitioned-3-SAT}
which was shown to be \NP-complete by Lampis et al.~\cite{DBLP:journals/siamdm/LampisMV25}.
In fact, we use a more restricted version of the problem, called \textsc{Exact-3-Partitioned-3-SAT},
which was shown to be \NP-complete by Foucaud et al.~\cite{DBLP:conf/icalp/FoucaudGK0IST24}.

\defproblem{\textsc{Exact-3-Partitioned-3-SAT}}{A formula $\psi$ in $3$-\textsc{CNF} form,
together with a partition of its variables into three
disjoint sets $X$, $Y$, and $Z$,
such that $|X| = |Y| = |Z| = n$, and
every clause contains exactly one variable from each of
$X$, $Y$, and $Z$.}{Determine whether $\psi$ is satisfiable.}

In Section~\ref{subsec:3SAT-to-SMD}, 
we reduce an instance \((\psi, \langle X, Y, Z\rangle)\) of \textsc{Exact-3-Partitioned-3-SAT} 
(with \(3n\) variables and \(m\) clauses) to an instance \((G,k)\) of \textsc{Strong Metric Dimension}. 
We establish the correctness of this reduction (Lemma~\ref{lemma:3SAT-to-SMD}) using an unconventional proof strategy, as outlined in Figure~\ref{fig:overview-strategy}. 

Rather than proving the correctness
of the reduction directly, 
we introduce a parallel reduction in Section~\ref{subsec:3SAT-to-VC}. 
Here, we map the same \textsc{Exact-3-Partitioned-3-SAT} instance to 
an equivalent \textsc{Vertex Cover} instance \((H, 3n + 2m)\). 
This adapts a standard reduction from Section~\ref{sec:np-hardness-diam}, 
making its correctness (Lemma~\ref{lemma:Ex-3-Par-3-SAT-to-VC}) straightforward to verify. 
Next, in Section~\ref{subsec:VC-equiv}, we construct the strong resolving graph \(G_{SR}\) of \(G\). 
By Proposition~\ref{prop:smd-vc}, solving \textsc{Strong Metric Dimension} on \(G\) is equivalent to solving \textsc{Vertex Cover} on \(G_{SR}\). 
The crux of our entire argument rests on a final technical lemma (Lemma~\ref{lemma:vertex-cover-two-graphs}), which bridges these parallel tracks by proving that \((G_{SR}, k)\) is a yes-instance of \textsc{Vertex Cover} if and only if \((H, 3n + 2m)\) is a yes-instance.

Combining these equivalences, we obtain that \((\psi, \langle X, Y, Z\rangle)\) is a 
yes-instance of \textsc{Exact-3-Partitioned-3-SAT} if and only if \((H, 3n + 2m)\) 
is a yes-instance of \textsc{Vertex Cover}, which holds if and only if 
\((G_{SR}, k)\) is a yes-instance of \textsc{Vertex Cover}. 
By Proposition~\ref{prop:smd-to-vc}, we have \(\smd(G) = \vc(G_{SR})\). 
Consequently, \((\psi, \langle X, Y, Z\rangle)\) is a yes-instance of \textsc{Exact-3-Partitioned-3-SAT} 
if and only if \((G, k)\) is a yes-instance of \textsc{Strong Metric Dimension}.

\begin{figure}[t]
\centering
\begin{tikzpicture}[
    box/.style={
        draw,
        rectangle,
        minimum width=2cm,
        minimum height=1cm,
        align=center
    },
    line/.style={
        draw,
        thick
    }
]

\node[box] (A) {%
Instance  $(\psi, \langle X, Y, Z\rangle)$ of \\
\textsc{Exact-3-Partitioned-3-SAT}
};

\node[box, right=2.7cm of A] (B) {%
Instance $(G,k)$ of\\
\textsc{Strong Metric Dimension}
};

\node[box, below=1.3cm of A] (C) {%
Instance $(H, 3n + 2m)$ of\\
\textsc{Vertex Cover}.
};

\node[box, below=1.3cm of B] (D) {%
Instance $(G_{SR}, k)$ of\\
\textsc{Vertex Cover}
};

\draw[line, <->, thick] (A) -- 
    node[midway, above]{Section~\ref{subsec:3SAT-to-SMD}} 
    node[midway, below]{Lemma~\ref{lemma:3SAT-to-SMD}} 
    (B);

\draw[line, <->, thick,dashed] (A) -- 
    node[midway, left, text width=1.6cm,
    align=left]{%
    Section~\ref{subsec:3SAT-to-VC}\\
    Lemma~\ref{lemma:Ex-3-Par-3-SAT-to-VC}}(C);

\draw[line, <->, thick,dashed] (B) -- 
    node[midway, right, text width=1.6cm,
    align=left]{%
    Definition~\ref{def:strong-resolving-graph}\\
    Proposition~\ref{prop:smd-vc}}
(D);

\draw[line, <->, thick,dashed] (C) -- 
    node[midway, above]{Section~\ref{subsec:VC-equiv}} 
    node[midway, below]{Lemma~\ref{lemma:vertex-cover-two-graphs}} 
(D);

\end{tikzpicture}
\caption{Overview of the strategy used in this section. Correctness of Lemma~\ref{lemma:3SAT-to-SMD}
follows from the correctness of Lemma~\ref{lemma:Ex-3-Par-3-SAT-to-VC}, 
Lemma~\ref{lemma:vertex-cover-two-graphs} and Proposition~\ref{prop:smd-vc}.\label{fig:overview-strategy}}
\end{figure}
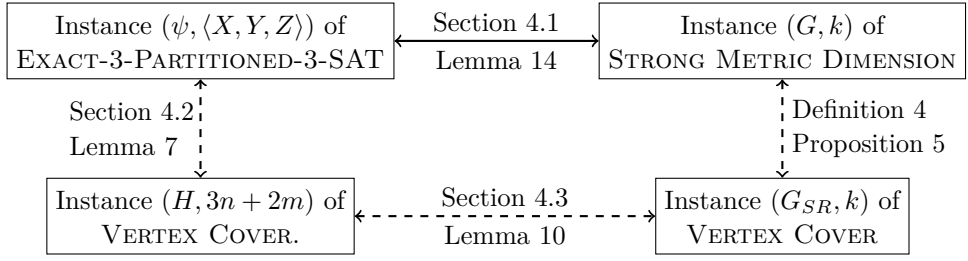

\subsection{Reduction to Strong Metric Dimension}\label{subsec:3SAT-to-SMD}

The reduction takes as input an instance \((\psi, \langle X, Y, Z\rangle)\) of 
\textsc{Exact-3-Partitioned-3-SAT} with $m$ clauses and $3n$ variables,
partitioned equally into sets $X$, $Y$, and $Z$.
Let \(X = \{x_1, x_2, \dots, x_n\}\),  
\(Y = \{y_1, y_2, \dots, y_n\}\), and
\(Z = \{z_1, z_2, \dots, z_n\}\).
The reduction outputs an instance $(G,k)$ of \textsc{Strong Metric Dimension}.
The graph $G$ consists of \emph{variable
encoding gadgets} and 
\emph{clause encoding gadgets} 
that are connected by long paths routed 
through intermediary vertices called 
\emph{portals}.

At a high level, every literal, both in 
its variable form and its occurrence in a 
clause, is represented by a \emph{critical vertex}.
We will ensure that a pair \((u, v)\) of non-pendant vertices is mutually maximally distant in \(G\) only if both
are critical vertices.
We now want to ensure that only the 
appropriate pair of critical vertices
are mutually maximally distant.
Consider a variable \(x_i\)
which appears positively in clause \(C_{\lambda}\).
On the variable side, the reduction creates a critical vertex \(x_{i,t}\) representing the positive literal of \(x_i\) whereas it creates a critical vertex \(x^\lambda_{i, t}\) on clause side representing the occurrence of the positive literal of \(x_i\) in clause \(C_{\lambda}\).
We want to ensure that \(x_{i,t}\) and \(x^\lambda_{i, t}\) are mutually maximally distant, but \(x_{i,t}\) is not mutually maximally distant from any other critical vertex representing a different literal
in a (different) clause, e.g., \(x^{\lambda'}_{i', t}\) for some \(i' \neq i\). 

The reduction adds the following four paths via portal
vertices $\alpha^{X_T,X_F}$ and $\beta^{X_T,X_F}$:
\begin{itemize}[nolistsep]
\item a path of length $(N^2 + iN + 1)$ connecting $x_{i,t}$ to $\alpha^{X_T,X_F}$;
\item a path of length $(N^2 - iN + 1)$ connecting $x_{i,t}$ to $\beta^{X_T,X_F}$;
\item a path of length $(N^2 + iN + 1)$ connecting $x^\lambda_{i, t}$ to $\alpha^{X_T,X_F}$; and
\item a path of length $(N^2 - iN + 1)$ connecting $x^\lambda_{i, t}$ to $\beta^{X_T,X_F}$.
\end{itemize}
This yields that the distance between $x_{i,t}$ and $x^\lambda_{i, t}$ is exactly $2N^2 + 2$.
We ensure that this distance is the largest possible distance between any two critical vertices in $G$ by introducing
a semi-global vertex $g$ and connecting it to all critical vertices via paths of length $N^2 + 1$.

Crucially, this perfect cancellation 
of \(\pm iN\) terms \emph{only} occurs when the index $i$ matches on both sides of the portal. 
If we attempt to cross-route paths between the positive literal of $x_i$ and 
the literal of a different variable $x_{i'}$ (where $i \neq i'$), 
the offsets do not cancel. 
The distance evaluates to 
$2N^2 \pm (i - i')N$, making
the shortest distance strictly less than $2N^2 + 2$.
In this case, the two corresponding
critical vertices are not mutually maximally distant.

We now move to the formal description of the reduction.
We partition the vertex set of the resulting graph $G$ into three disjoint parts:
\begin{enumerate}[nolistsep]
\item vertices of degree exactly one, called \emph{pendant vertices};
\item vertices adjacent to at least one pendant vertex, called \emph{support vertices}; and
\item non-pendant vertices that are not adjacent to any pendant vertex, called \emph{critical vertices}.
\end{enumerate}
We denote these sets by $V_p$, $V_s$, and $V_c$, respectively.  
Initially, the reduction sets $V(G)=V_p=V_s=V_c=\emptyset$.
For notational convenience, we define two functions $\g, \f : V_c \to V_s$.
Whenever we say \emph{add a vertex $u$ to $V_s$}, we also add a new
pendant vertex $p$ to $V_p$ and the edge $up$ to $E(G)$, without stating this
explicitly each time.  
Moreover, once a vertex $u$ is added to $V_c$, we ensure that no pendant
vertices will be attached to $u$ later in the construction.
See Figure~\ref{fig:reduction-overview} for an overview of the construction.

\subparagraph{Encoding Variables.}
For every variable $x_i \in X$, the reduction adds two vertices
$x_{i,t}$ and $x_{i,f}$ to $V_c$, representing the positive and negative literals,
respectively.  It also adds four vertices 
$\g(x_{i,t}), \f(x_{i,t}), \g(x_{i,f})$, and $\f(x_{i,f})$ to $V_s$.
It adds edges so that $x_{i,t}$ is adjacent to $\g(x_{i,t})$ and $\f(x_{i,t})$, and
$x_{i,f}$ is adjacent to $\g(x_{i,f})$ and $\f(x_{i,f})$.

Define \(X_T = \{x_{i,t} \mid i \in [n]\}\) and \(X_F = \{x_{i,f} \mid i \in [n]\}\).
We construct an analogous gadget for every variable in $Y$ and $Z$,
defining $Y_T, Y_F, Z_T$, and $Z_F$ in the same way.

\begin{figure}[!htbp]
    \centering
\includegraphics[scale=0.95]{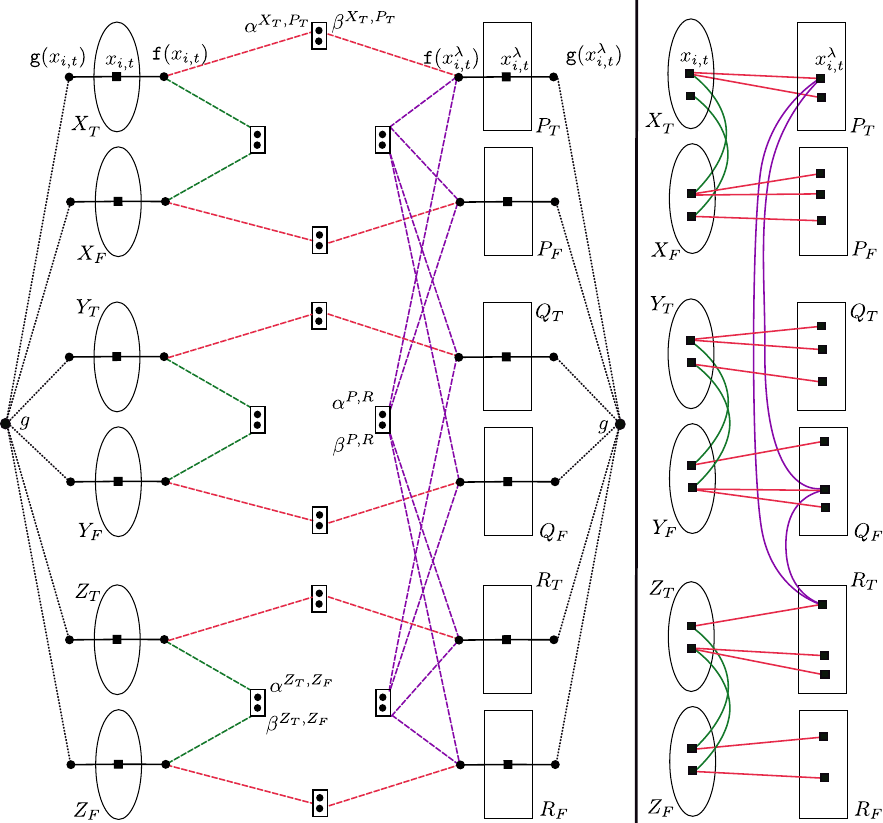}
\caption{(Left) Overview of the graph $G$ constructed by the reduction from
\textsc{Exact-3-Partitioned-3-SAT} to \textsc{Strong Metric Dimension}.
For notational clarity, we show the same vertex $g$ twice (the vertices in
the middle-left and middle-right of the figure represent the same vertex).
Moreover, we label only a subset of the vertices; all vertices belonging to
the same set have the same type of connections as the representative vertex
shown in the figure.
Vertices marked by $(\blacksquare)$ are critical vertices.
All remaining vertices, denoted by $(\bullet)$ or lying on the paths drawn
with dotted or dashed lines, are support vertices.
Pendant vertices are omitted from the figure for clarity.
Each dotted path starting from $g$ has length $N^2$.
The lengths of the dashed paths depend on their endpoints and are of the form
$N^2 \pm iN$ for some $i \in [n]$ or $i \in [m]$.
(Right) Overview of the graph $H$ constructed by the reduction from
\textsc{Exact-3-Partitioned-3-SAT} to \textsc{Vertex Cover}.\label{fig:reduction-overview}}
\end{figure}

\subparagraph{Encoding Clauses.}
Consider a clause 
\(C_{\lambda} = \langle x_i \lor \neg y_j \lor z_{\ell} \rangle \)
for some $\lambda \in [m]$ and \(i, j, \ell \in [n]\).  
The reduction adds three vertices 
$x^{\lambda}_{i,t}$, $y^{\lambda}_{j,f}$, and $z^{\lambda}_{\ell,t}$ to $V_c$,
corresponding to the literals appearing in $C_{\lambda}$.  
As in the previous step, for every $u \in \{x^{\lambda}_{i,t}, y^{\lambda}_{j,f}, z^{\lambda}_{\ell,t}\}$,
the reduction adds $\g(u)$ and $\f(u)$ to $V_s$ and makes both adjacent to $u$.

Let $P_T$ denote the collection of vertices corresponding to
positive occurrences of variables from $X$ in clauses.
Formally,
\(P_T = \{ x^{\lambda}_{i,t} \mid i \in [n], \lambda \in [m] \text{ and } x_i \text{ appears positively in } C_{\lambda} \}\).
Note that a variable $x_i$ may have multiple representatives in $P_T$,
depending on how many clauses contain $x_i$ positively.
Analogously, define $P_F$ for negative occurrences of variables in $X$.
Similarly, define $Q_T$ and $Q_F$ for positive and negative occurrences
of variables in $Y$, and define $R_T$ and $R_F$ for positive and negative
occurrences of variables in $Z$.

\subparagraph{Semi-global Vertex and Connecting Paths.}
The reduction adds a vertex $g$ to $V_s$, which serves as a semi-global vertex.
Let $N = {\left( n + m \right)}^2$.
For every vertex of the form $\g(u)$ with $u \in V_c$, the reduction adds
a path of length $N^2$ connecting $g$ to $\g(u)$.  
Every internal vertex of this path is added to $V_s$.
Consequently, all vertices in $V_c$ are at distance at most $2N^2 + 2$ from one another.

\subparagraph{Portal Vertices and Connecting Paths.}
A \emph{portal} is a pair of vertices denoted by $\alpha^{A,B}$ and $\beta^{A,B}$,
indicating that the portal connects only to vertices in the sets $A$ and $B$.
Alternately, it would be helpful to think that these two vertices
\emph{ports} the paths starting at vertices in \(A\) to vertices in \(B\) (and vice versa). 
Every vertex introduced on a portal path, 
including vertices in portals, belongs to $V_s$
(and hence receives an attached pendant vertex).
The reduction introduces three types of portals, depending on the sets being connected.
See Figure~\ref{fig:distance-overview}.
\begin{enumerate}[nolistsep]
\item \emph{Connecting positive and negative literals (which are indicated in green in Figure~\ref{fig:distance-overview}).}
Consider the sets $X_T$ and $X_F$.
The reduction adds a portal consisting of vertices
$\alpha^{X_T,X_F}$ and $\beta^{X_T,X_F}$.
For each variable $x_i \in X$, the reduction adds the following four paths:
\begin{itemize}
\item a path of length $(N^2 + iN)$ connecting $\f(x_{i,t})$ to $\alpha^{X_T,X_F}$;
\item a path of length $(N^2 - iN)$ connecting $\f(x_{i,t})$ to $\beta^{X_T,X_F}$;
\item a path of length $(N^2 - iN)$ connecting $\f(x_{i,f})$ to $\alpha^{X_T,X_F}$;
\item a path of length $(N^2 + iN)$ connecting $\f(x_{i,f})$ to $\beta^{X_T,X_F}$.
\end{itemize}

The reduction adds analogous portals
$\alpha^{Y_T,Y_F}, \beta^{Y_T,Y_F}$ and
$\alpha^{Z_T,Z_F}, \beta^{Z_T,Z_F}$,
together with corresponding paths from $\f(u)$ for all
$u \in Y_T \cup Y_F$ and $u \in Z_T \cup Z_F$, respectively.

\item \emph{Connecting variables to their clause occurrences (which are indicated in red in Figure~\ref{fig:distance-overview}).}
Consider the set $X_T$ of positive literals of variables in $X$
and the set $P_T$ of positive clause occurrences of variables in $X$.
The reduction introduces a portal
$\alpha^{X_T,P_T}$ and $\beta^{X_T,P_T}$.

\begin{itemize}[nolistsep]
\item For every $x_{i,t} \in X_T$, the reduction adds:
\begin{itemize}[nolistsep]
\item a path of length $(N^2 + iN)$ from $\f(x_{i,t})$ to $\alpha^{X_T,P_T}$;
\item a path of length $(N^2 - iN)$ from $\f(x_{i,t})$ to $\beta^{X_T,P_T}$.
\end{itemize}

\item For every $x^{\lambda}_{i,t} \in P_T$, the reduction adds:
\begin{itemize}[nolistsep]
\item a path of length $(N^2 - iN)$ from $\f(x^{\lambda}_{i,t})$ to
      $\alpha^{X_T,P_T}$;
\item a path of length $(N^2 + iN)$ from $\f(x^{\lambda}_{i,t})$ to
      $\beta^{X_T,P_T}$.
\end{itemize}
\end{itemize}

The reduction constructs analogous portals and paths for the pairs:
\((X_F,P_F)\), 
\((Y_T,Q_T)\),
\((Y_F,Q_F)\), 
\((Z_T,R_T)\),
\((Z_F,R_F)\).

\begin{figure}[!tbp]
\begin{center}
\includegraphics[scale=1]{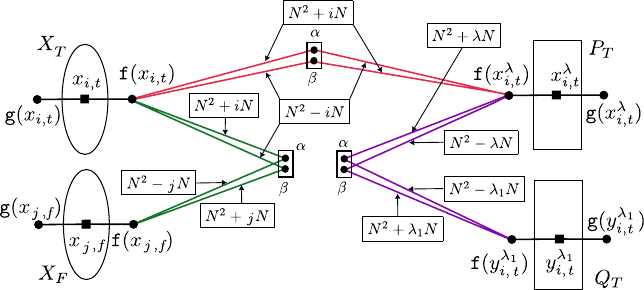}
\end{center}
\caption{Overview of the distances in graph \(G\). For the notation clarity, 
we do not specify the complete name of the portal vertices.
{We add a negative literal with index \(j\) to highlight the fact that \(x_{i, t}\) and \(x_{j, f}\) are at distance
\(2N^2\) if and only if \(i = j\).}\label{fig:distance-overview}}
\end{figure}

\item \emph{Connecting literals in the same clause (which are indicated in purple in Figure~\ref{fig:distance-overview}).}
Define 
\(P = P_T \cup P_F\), \(Q = Q_T \cup Q_F\), \(R = R_T \cup R_F\).
The reduction introduces the following three portals:
\((\alpha^{P,Q}, \beta^{P,Q})\), \((\alpha^{Q,R}, \beta^{Q,R})\),
\((\alpha^{R,P}, \beta^{R,P})\).
Consider a clause 
\(C_{\lambda} = \langle x_i \lor \neg y_j \lor z_{\ell} \rangle \)
for some $\lambda \in [m]$.  
By construction, this clause corresponds to vertices
$x^{\lambda}_{i,t} \in P_T$, 
$y^{\lambda}_{j,f} \in Q_F$, and 
$z^{\lambda}_{\ell,t} \in R_T$.
For this clause, the reduction adds four paths for each of the following three pairs of vertices:
\((i)\) $x^{\lambda}_{i,t}$ and $y^{\lambda}_{j,f}$,
\((ii)\) $y^{\lambda}_{j,f}$ and $z^{\lambda}_{\ell,t}$, and
\((iii)\) $z^{\lambda}_{\ell,t}$ and $x^{\lambda}_{i,t}$.
The lengths of these paths are specified in Table~\ref{table:lenghts-clause-to-clause}.
\end{enumerate}

\begin{table}[t]
\centering
\begin{minipage}{0.32\textwidth}
\centering
\renewcommand{\arraystretch}{1.5}
\begin{tabular}{p{0.17\textwidth}p{0.18\textwidth}|p{0.42\textwidth}}
\toprule
\multicolumn{2}{c}{Endpoints} & Length \\
\midrule
\rowcolor{gray!20}\(\f(x^{\lambda}_{i, t})\) & \(\alpha^{P, Q}\) & \((N^2 + \lambda \cdot N)\) \\
\midrule
\rowcolor{gray!20}\(\f(x^{\lambda}_{i, t})\) & \(\beta^{P, Q}\) & \((N^2 - \lambda \cdot N)\) \\
\midrule
\(\f(x^{\lambda_1}_{i, f})\) & \(\alpha^{P, Q}\) & \((N^2 + \lambda_1 \cdot N)\) \\
\midrule
\(\f(x^{\lambda_1}_{i, f})\) & \(\beta^{P, Q}\) & \((N^2 - \lambda_1 \cdot N)\) \\
\midrule
\(\f(y^{\lambda_2}_{j, t})\) & \(\alpha^{P, Q}\) & \((N^2 - \lambda_2 \cdot N)\) \\
\midrule
\(\f(y^{\lambda_2}_{j, t})\) & \(\beta^{P, Q}\) & \((N^2+\lambda_2\cdot N)\) \\
\midrule
\rowcolor{gray!20}\(\f(y^{\lambda}_{j, f})\) & \(\alpha^{P, Q}\) & \((N^2-\lambda \cdot N)\) \\
\midrule
\rowcolor{gray!20}\(\f(y^{\lambda}_{j, f})\) & \(\beta^{P, Q}\) & \((N^2+\lambda \cdot N)\) \\
\bottomrule
\end{tabular}
\end{minipage}
\hfill
\begin{minipage}{0.32\textwidth}
\centering
\renewcommand{\arraystretch}{1.5}
\begin{tabular}{p{0.17\textwidth}p{0.18\textwidth}|p{0.42\textwidth}}
\toprule
\multicolumn{2}{c}{Endpoints} & Length \\
\midrule
\(\f(y^{\lambda_3}_{j, f})\) & \(\alpha^{Q, R}\) & \((N^2 + \lambda_3 \cdot N)\) \\
\midrule
\(\f(y^{\lambda_3}_{j, f})\) & \(\beta^{Q, R}\) & \((N^2 - \lambda_3 \cdot N)\) \\
\midrule
\rowcolor{gray!20}\( \f(y^{\lambda}_{j, f})\) & \(\alpha^{Q, R}\) & \((N^2 + \lambda \cdot N)\) \\
\midrule
\rowcolor{gray!20}\( \f(y^{\lambda}_{j, f})\) & \(\beta^{Q, R}\) & \((N^2 - \lambda \cdot N)\) \\
\midrule
\rowcolor{gray!20}\( \f(z^{\lambda}_{\ell, t})\) & \(\alpha^{Q, R}\) & \((N^2 - \lambda \cdot N)\) \\
\midrule
\rowcolor{gray!20}\( \f(z^{\lambda}_{\ell, t})\) & \(\beta^{Q, R}\) & \((N^2 + \lambda \cdot N)\) \\
\midrule
\(\f(z^{\lambda_4}_{\ell, f})\) & \(\alpha^{Q, R}\) & \((N^2 - \lambda_4 \cdot N)\) \\
\midrule
\(\f(z^{\lambda_4}_{\ell, f})\) & \(\beta^{Q, R}\) & \((N^2 + \lambda_4 \cdot N)\) \\
\bottomrule
\end{tabular}
\end{minipage}
\hfill
\begin{minipage}{0.32\textwidth}
\centering
\renewcommand{\arraystretch}{1.5}
\begin{tabular}{p{0.16\textwidth}p{0.18\textwidth}|p{0.42\textwidth}}
\toprule
\multicolumn{2}{c}{Endpoints} & Length \\
\midrule
\rowcolor{gray!20}\( \f(z^{\lambda}_{\ell, t})\) & \(\alpha^{R, P}\) & \((N^2 + \lambda \cdot N)\) \\
\midrule
\rowcolor{gray!20}\( \f(z^{\lambda}_{\ell, t})\) & \(\beta^{R, P}\) & \((N^2 - \lambda \cdot N)\) \\
\midrule
\(\f(z^{\lambda_5}_{\ell, f})\) & \(\alpha^{R, P}\) & \((N^2 + \lambda_5 \cdot N)\) \\
\midrule
\(\f(z^{\lambda_5}_{\ell, f})\) & \(\beta^{R, P}\) & \((N^2 - \lambda_5 \cdot N)\) \\
\midrule
\rowcolor{gray!20}\( \f(x^{\lambda}_{i, t})\) & \(\alpha^{R, P}\) & \((N^2 - \lambda \cdot N)\) \\
\midrule
\rowcolor{gray!20}\( \f(x^{\lambda}_{i, t})\) & \(\beta^{R, P}\) & \((N^2 + \lambda \cdot N)\)\\
\midrule
\(\f(x^{\lambda_6}_{i, f})\) & \(\alpha^{R, P}\) & \((N^2 - \lambda_6 \cdot N)\) \\
\midrule
\(\f(x^{\lambda_6}_{i, f})\) & \(\beta^{R, P}\) & \((N^2 + \lambda_6 \cdot N)\)\\
\bottomrule
\end{tabular}
\end{minipage}
\caption{The paths added by the reduction corresponding to clause 
\(C_{\lambda} = \langle x_i \lor \neg y_j \lor z_{\ell} \rangle\)
are in the rows highlighted by gray rows.
The other rows correspond to negation of the literals
appeared in \(C_{\lambda}\), assuming they are 
appearing in other clause.
The three tables corresponds to the connections between
\(P_T \cup P_F\) and \(Q_T \cup Q_F\), 
\(Q_T \cup Q_F\) and \(R_T \cup R_F\), 
and
 \(R_T \cup R_F\) and \(P_T \cup P_F\), respectively.\label{table:lenghts-clause-to-clause}}
 \vspace{-0.75cm}
\end{table}

This completes the construction of graph \(G\).
Recall that the number of variables and clauses in the instance \(\psi\)
of \textsc{Exact-3-Partitioned-3-SAT} are \(3n\) and \(m\),
respectively. 
Moreover,  \(V_p\) is the collection 
of all the pendant vertices 
in \(G\).
The reduction defines \(k = 3n + 2m + |V_p| - 1\),
and returns \((G, k)\) as an instance of \textsc{Strong Metric Dimension}.

\subparagraph{Properties of Graph \(G\).}
Let \(S\) denote the set of all portal vertices introduced in the construction.
There are three portals to connect 
positive and negative literals, 
six portals to connect variables 
to their clause occurrences, 
and three portals to connect literals within the same clause.
Hence, we have \(12\) portals,
which implies \(24\) portal endpoints, i.e. \(|S| = 24\). 
Define \(S_g = S \cup \{g\}\) where
\(g\) is the semi-global vertex. 

We claim that \(S_g\) forms a feedback vertex set of \(G\). 
Consider the graph \(G - S_g\). 
Removing the portals and the semi-global vertex deletes all long-distance
cycles connecting the local gadgets. 
The remaining graph is a 
collection of pairwise 
disjoint connected components, 
each corresponding to a single literal 
in a variable encoding or a clause encoding, or to pendant vertices attached to the portal vertices. 
Within any such component, the structure 
consists only of a critical vertex, its 
associated support and pendant vertices, and the paths which now terminate at
a pendant vertex.
Consequently, \(G - S_g\) is a forest, which implies that the feedback vertex set number \(\fvs\) of the graph is at most \(|S_g| = 25\).

Furthermore, note that each connected component of \(G - S_g\) is a tree obtained by subdividing the edges of a star with at most five leaves, where the central vertex corresponds to the critical vertex representing a literal, and the leaves correspond to the support vertices and the paths attached to it. Therefore, each connected component of \(G - S_g\) has pathwidth at most \(2\). 
Using the fact that \(\pw(G) \leq |S_g| + \pw(G - S_g)\), we have \(\pw(G) \leq 27\).

\subsection{Reduction to Vertex Cover}\label{subsec:3SAT-to-VC}

Recall the standard textbook reduction from 
\textsc{$3$-SAT} to \textsc{Vertex Cover} 
~\cite{DBLP:books/daglib/0015106}, which we have also discussed in
Section~\ref{sec:np-hardness-diam}.  
We adapt this to the reduction that given an instance $(\psi, \langle X, Y, Z \rangle)$ of 
\textsc{Exact-3-Partitioned-3-SAT} with $3n$ variables and $m$ clauses,
constructs a graph $H$ as follows.
By construction, \(V(H)\) is a subset of \(V(G)\), 
and hence we use the same names.
See Figure~\ref{fig:reduction-overview}.

\subparagraph{Encoding Variables.}
For every variable $x_i \in X$, the reduction adds two vertices
$x_{i,t}$ and $x_{i,f}$ to $H$, representing the positive and negative literals,
respectively.  
It also adds a matching edge $x_{i,t}x_{i,f}$ for every $i \in [n]$.
These edges are shown in green in Figure~\ref{fig:reduction-overview}.
We define the sets $X_T, X_F, Y_T, Y_F, Z_T$, and $Z_F$ exactly as in the previous reduction.

\subparagraph{Encoding Clauses.}
Consider a clause 
\(C_{\lambda} = \langle x_i \lor \neg y_j \lor z_{\ell} \rangle\)
for some $\lambda \in [m]$ and \(i, j, \ell \in [n]\).  
The reduction adds three vertices 
$x^{\lambda}_{i,t}$, $y^{\lambda}_{j,f}$, and $z^{\lambda}_{\ell,t}$ to $H$,
corresponding to the literals appearing in $C_{\lambda}$.  
It adds edges to make these three vertices adjacent with each other.
These edges are indicated in purple in Figure~\ref{fig:reduction-overview}.
We define the sets $P_T, P_F, Q_T, Q_F, R_T$, and $R_F$ as in the previous reduction.

\subparagraph{Connecting literals and their appearances in clauses.}
Consider the set $X_T$ of positive literals of variables in $X$
and the set $P_T$ of their positive clause occurrences.
For every $x_{i,t} \in X_T$ and every corresponding vertex
$x^{\lambda}_{i,t} \in P_T$ for some $\lambda \in [m]$,  
the reduction adds the edge $x_{i,t}x^{\lambda}_{i,t}$.  
These edges are shown in red in Figure~\ref{fig:reduction-overview}.
The reduction adds analogous edges between the following pairs of sets:
\((X_F, P_F)\), \((Y_T, Q_T)\),\((Y_F, Q_F)\),\((Z_T, R_T)\), and  \((Z_F, R_F)\).

The reduction outputs $(H,k)$ as the resulting instance of 
\textsc{Vertex Cover}, where $k = 3n + 2m$.
The correctness of the following lemma follows from the standard reduction 
mentioned in~\cite[Section~3.1.3]{DBLP:books/fm/GareyJ79}.

\begin{lemma}\label{lemma:Ex-3-Par-3-SAT-to-VC}
For a formula $\psi$ with $3n$ variables and $m$ clauses,
\((\psi, \langle X, Y, Z\rangle)\) is a yes-instance of \textsc{Exact-3-Partitioned-3-SAT}
if and only if  \((H, 3n + 2m)\) is a yes-instance of \textsc{Vertex Cover}.
\end{lemma}

\subsection{\texorpdfstring{Strong Resolving Graph \(G_{SR}\)} {Strong Resolving Graph GSR}}\label{subsec:VC-equiv}

In this subsection, we consider the strong resolving graph \(G_{SR}\) of \(G\)
and compare its vertex cover with that of \(H\).
Note that, by the construction, \(V(H)\) is a subset of \(V(G)\) and hence also of \(V(G_{SR})\).
In fact, by the construction, \(V(H) = V_c\).
We now argue about the structure of graph \(G_{SR}\).
Towards that, we recall the following easy, but critical, 
observation in~\cite{DBLP:conf/icalp/FoucaudGK0IST24}.
See~\cite[Observation~22]{DBLP:conf/icalp/FoucaudGK0IST24}.

\begin{observation}\label{obs:pendant-in-G-to-clique-in-GSR}
Consider a connected graph $G$ that has at least $3$ vertices.
Suppose $P \subseteq V(G)$ is the collection of all the pendant vertices in $G$.
Then, $P$ is a clique in $G_{SR}$, and every vertex in $N(P)$ is an isolated
vertex in $G_{SR}$.
\end{observation}


By the construction in Subsection~\ref{subsec:3SAT-to-SMD}, and especially
because of the semi-global vertex \(g\) and its connection to other vertices,
it is evident that \(G\) is a connected graph.
Hence, the above observation is applicable to \(G\).
Recall that \(V_p\) is the collection of pendant vertices in \(G\), 
\(V_s\) is the collection of support vertices, i.e.~vertices that are adjacent to a pendant vertex,
where the set of remaining set of vertices, called critical vertices, is denoted by \(V_c\).

By Definition~\ref{def:strong-resolving-graph}, and the above observation, we can
conclude that \(V(G_{SR})\) can be partitioned into the following three parts:
\begin{itemize}[nolistsep]
\item set of vertices \(V_p\) which forms a clique,
\item set of vertices \(V_s\) which is a collection of isolated vertices, and
\item the remaining set of vertices denoted by \(V_c\).
\end{itemize}

We first prove the following observation.

\begin{observation}\label{obs:wlog-vertex-cover}
    There exists a vertex \(p \in V_p\) such that, without loss of generality, 
    every vertex cover 
    \(C \subseteq V(G_{SR})\) 
    contain all vertices in \(V_p \setminus \{p\}\).
\end{observation}
\begin{claimproof}
We argue that there exists a vertex \(p \in V_p\) such that \(p\) is not adjacent to any vertex in \(V_c\) in \(G_{SR}\).
Consider the graph \(G\) and let \(p\) be the pendant vertex adjacent to the vertex \(g\).
Let \(u\) be an arbitrary vertex in \(V_c\). 
There are two types of paths connecting \(p\) to \(\g(u)\) (and similarly to \(\f(u)\)):
paths that do not contain any portal vertices, and
paths that pass through at least one portal vertex.
By the construction of \(G\), every path from \(p\) to either \(\g(u)\) or \(\f(u)\) that contains 
a portal vertex is strictly longer than a path that avoids all portal vertices. 
Hence, the shortest paths from \(p\) to \(\g(u)\) and from \(p\) to \(\f(u)\) are precisely the paths that avoid portal vertices.
Consequently, \(\dist_G(p, \g(u)) = N^2 + 1\) and \(\dist_G(p, \f(u)) = N^2 + 3\).
Hence, \(u\) is not maximally distant from \(p\) in \(G\).
Therefore, \(u\) is not adjacent to \(p\) in \(G_{SR}\).
Since \(u\) is an arbitrary vertex in  \(V_c\), we conclude that there exists a vertex 
\(p \in V_p\) that is not adjacent to any vertex in \(V_c\) in \(G_{SR}\).
It follows that, we can assume without loss of generality that
every vertex cover \(C \subseteq V(G_{SR})\) contain all vertices in \(V_p \setminus \{p\}\).
\end{claimproof}

We now argue that a set \(C \subseteq V(G_{SR})\) is a vertex cover of \(G_{SR}\) if and only if 
\(C \setminus (V_p \setminus \{p\})\) 
is a vertex cover of \(H\). Formally, we prove the following lemma.

\begin{lemma}\label{lemma:vertex-cover-two-graphs}
Consider set \(C \subseteq V(G_{SR})\) that contains \(V_p \setminus \{p\}\).
Set \(C\) is a vertex cover of \(G_{SR}\) if and only if \(C \setminus (V_p \setminus \{p\})\)
is a vertex cover of \(H\).
\end{lemma}

We dedicate rest of the section to prove the above lemma.
In fact, we prove a stronger statement. 
We prove that graphs \(H\) and \(G_{SR}[V_c]\) are identical. 
This, and the facts that \(V_s\) are isolated vertices in \(G_{SR}\)
and every vertex cover of \(G_{SR}\) contains \(V_p \setminus \{p\}\)
proves the above lemma.

By the construction, it is easy to see that \(V(H) = V(G) \setminus (V_p \cup V_s)\).
We argue that \(E(H) = E(G_{SR}[V_c])\).
In order to prove this, we argue that for any two vertices \(u, v\) in \(H\),
edge \(uv\) is in \(H\) if and only if vertices \(u, v\) are mutually 
maximally distant in \(G\). See Definition~\ref{def:mutually-max-distant}.

Before proceeding, recall that \(V_c\) is partitioned
into the following twelve parts:
\(X_T\), \(X_F\), \(Y_T\), \(Y_F\), \(Z_T\), 
\(Z_F\),
\(P_T\), \(P_F\), \(Q_T\), \(Q_F\), \(R_T\), and \(R_F\).
These parts are highly symmetric, and hence, we can argue about the presence or absence of edges across these parts 
by considering the cases mentioned in Table~\ref{tabel:exhaustive-cases}.
We defined \(S\) as the collection of vertices in \(V(G)\) which are in portals. 
As mentioned before \(G - (S \cup \{g\})\) 
is the collection of disjoint subdivided stars, whose center corresponds to vertices added while encoding variables or clauses or pendant vertices attached to portal vertices.
This implies in graph \(G - \{g\}\), any path connecting variable or clause gadgets must pass through portal vertices \(S\).
Hence, the following cases 
are exhaustive for any pair of vertices
\(u, v\) in \(V_c\):
\begin{enumerate}[nolistsep]
\item Every path from \(u\) to \(v\) in 
\(G - \{g\}\) contains at least two vertices from (different) portals.
We handle this case in Claim~\ref{claim:two-portals-in-paths} and show that \(u\) and \(v\) are not mutually maximally distant in \(G\).
This implies that there is no edge with endpoints \(u, v\) in \(G_{SR}\).
Note that by construction of \(H\), there is no edge with endpoints \(u, v\).
\item  Every path from \(u\) to \(v\) in \(G - \{g\}\) contains exactly one vertex from a portal.
In Claim~\ref{claim:Xt-Xt} we handle the case when \(u\) and \(v\) are in the same part, and show that \(u\) and \(v\) are mutually maximally distant in \(G\).
This implies that there is an edge with endpoints \(u, v\) in \(G_{SR}\).
Note that by construction of \(H\), every
part is an independent set.
We handle the remaining cases in Claim~\ref{claim:edges-in-H} and show that green, red, and purple edges in \(H\) are precisely the edges in \(G_{SR}\).
\end{enumerate}

\begin{table}[t]
\centering
\begin{tabular}{p{0.03\textwidth}|p{0.14\textwidth}p{0.14\textwidth}p{0.1\textwidth}p{0.14\textwidth}p{0.13\textwidth}p{0.13\textwidth}}
\toprule
 & \(X_T\) & \(X_F\) & \(Y_T/Y_F\)/ \(Z_T/Z_F\) & \(P_T\) & \(P_F\) & \(Q_T/Q_F\)/ \(R_T/R_F\) \\ 
\midrule
\(X_T\) & \cellcolor{gray!20} Claim~\ref{claim:Xt-Xt} & Claim~\ref{claim:edges-in-H}(\ref{item:green-edges})  & \cellcolor{gray!20} Claim~\ref{claim:two-portals-in-paths}  & Claim~\ref{claim:edges-in-H}(\ref{item:red-edges}) & \cellcolor{gray!20} Claim~\ref{claim:two-portals-in-paths} & \cellcolor{gray!20} Claim~\ref{claim:two-portals-in-paths} \\
\midrule
\(P_T\) & Claim~\ref{claim:edges-in-H}(\ref{item:red-edges})  & \cellcolor{gray!20} Claim~\ref{claim:two-portals-in-paths} & \cellcolor{gray!20} Claim~\ref{claim:two-portals-in-paths} & \cellcolor{gray!20} Claim~\ref{claim:Xt-Xt} & \cellcolor{gray!20} Claim~\ref{claim:edges-in-H}(\ref{item:no-edges}) & Claim~\ref{claim:edges-in-H}(\ref{item:purple-edges}) \\
\bottomrule
\end{tabular}
\caption{The entry in the cell arguments about the presence or absence of 
edges across the set corresponding to its row and column.
Sets \(X_T\) and \(P_T\) are representative of the sets on variable side
and of the sets on clause side respectively.
The gray cells denote the absence of edges
across the corresponding sets.\label{tabel:exhaustive-cases}}
\vspace{-0.75cm}
\end{table}


\begin{claim}\label{claim:two-portals-in-paths} 
Consider two vertices \(u\) and \(v\) in \(V_c\) such that 
every path from \(u\) to \(v\) in \(G - \{g\}\) contains
at least two vertices from (different) portals. 
Then \(u\) and \(v\) are \emph{not} mutually maximally distant in \(G\).
\end{claim}
\begin{claimproof} For notational convenience, we present the argument for the case when \(u \in X_T\) and \(v \in P_F\); the remaining cases follow analogously.

Let \(u = x_{i,t}\) and \(v = x^{\lambda}_{j,f}\), where \(i,j \in [n]\) and \(\lambda \in [m]\).
By construction, every path in \(G - \{g\}\) between \(x_{i,t}\) and \(x^{\lambda}_{j,f}\) contains at least two portal vertices belonging to distinct portal gadgets.
Any such path must traverse segments whose lengths are at least of the form 
\(N^2 - iN\) and \(N^2 - \lambda N\).
Since \(i \le n\) and \(\lambda \le m\), we have
\(N^2 - iN \ge N^2 - (n+m)N\) and \(N^2 - \lambda N \ge N^2 - (n+m)N\).
Consequently, every path between 
\(x_{i,t}\) and \(x^{\lambda}_{j,f}\) in \(G - \{g\}\) has length at least \(4N^2 - 4(n+m)N\).

Now consider the path in \(G\) from \(x_{i,t}\) to \(\g(x_{i,t})\) to \(g\) to \(\g(x^{\lambda}_{j,f})\)
which finally reaches \(x^{\lambda}_{j,f}\).
By construction, the total length of this path is \(2N^2 + 2\).
Recall that \(N = \left( n+m \right)^2\).
Therefore, any \(n+m \ge 2\), we have
\(2N^2 + 2 < 4N^2 - 4(n+m)N\),
and hence the shortest path between \(x_{i,t}\) and \(x^{\lambda}_{j,f}\) in \(G\) 
necessarily uses the vertex \(g\).
In particular, \(\dist_G(x_{i,t}, \g(x^{\lambda}_{j,f})) 
< \dist_G(x_{i,t}, x^{\lambda}_{j,f}) 
< \dist_G(x_{i,t}, \f(x^{\lambda}_{j,f}))\).
It follows that \(v = x^{\lambda}_{j,f}\) is not maximally distant from \(u = x_{i,t}\).
Hence, \(u\) and \(v\) are not mutually maximally distant in \(G\).

The same reasoning applies to any other pair of vertices satisfying 
the hypothesis of the claim.
\end{claimproof}

In the next claim, we argue about the 
non-existence of an edge
within the same part.

\begin{claim}\label{claim:Xt-Xt} 
Consider two vertices \(u\) and \(v\) that are in the same part of \(V_c\)
(and hence there exists a path from \(u\) to \(v\) in \(G - \{g\}\) that contains exactly 
one portal vertex).
Then, \(u\) and \(v\) are \emph{not} mutually maximally distant in \(G\).
\end{claim}
\begin{claimproof}
We first treat the case \(u, v \in X_T\), and then explain how the argument extends to any part of \(V_c\).

Assume that \(u, v \in X_T\). Then \(u = x_{i,t}\) and \(v = x_{j,t}\) for some \(i \neq j \in [n]\).
By the construction of \(H\), the vertices \(x_{i,t}\) and \(x_{j,t}\) are non-adjacent in \(H\).
Hence, it suffices to prove that \(x_{i,t}\) and \(x_{j,t}\) are not maximally distant from each other.

Consider any path starting at \(x_{i,t}\) and ending at \(x_{j,t}\).
By construction, every such path contains either a portal vertex in \(S\) or the semi-global vertex \(g\).
Recall that \(\g(x_{i,t})\) and \(\f(x_{i,t})\) are the only two neighbors of \(x_{i,t}\), and the analogous 
statement holds for \(x_{j,t}\).
Moreover, both \(\f(x_{i,t})\) and \(\f(x_{j,t})\) are adjacent to portal vertices in 
\(\{\alpha^{X_T,P_T}, \beta^{X_T,P_T}\}\) and 
\(\{\alpha^{X_T,X_F}, \beta^{X_T,X_F}\}\).
Consider a path that connects \(x_{i,t}\) to \(x_{j,t}\) via
\(\f(x_{i,t})\), then a portal vertex \(\alpha^{X_T, P_T}\)(or \(\beta^{X_T, P_T}\)), and
subsequently via \(\f(x_{j,t})\).
By construction, the length of this path is \(2N^2 - (i + j)N + 2\).
Now consider the alternative path from \(x_{i,t}\) to \(x_{j,t}\) that proceeds via
\(\g(x_{i,t})\), then the semi-global vertex \(g\), and subsequently via \(\g(x_{j,t})\).
The length of this path is \(2N^2 + 2\).
Since \(i, j \in [n]\), we have \(2N^2 - (i + j)N + 2 < 2N^2 + 2\).
Therefore, for every \(i, j \in [n]\), an isometric path between \(x_{i,t}\) and \(x_{j,t}\) has 
length \(2N^2 - (i + j)N + 2\) and passes through a portal vertex.
In particular, \(\dist_G(x_{i,t}, \f(x_{j,t}))  <
\dist_G(x_{i,t}, x_{j,t}) < \dist_G(x_{i,t}, \g(x_{j,t}))\).
Hence, \(u = x_{j,t}\) is not maximally distant from \(v = x_{i,t}\).

We now explain how this argument generalizes to any part of \(V_c\).
In the above reasoning, we used that when \(u, v \in X_T\), the distance from a 
variable-encoding vertex to at least one portal vertex is strictly smaller than \(N^2 + 1\).
By construction, this property holds for every part of \(V_c\).
Formally, for any part of \(V_c\), there exists at least one portal vertex (either of \(\alpha\)-type or \(\beta\)-type)
adjacent to that part whose distance from the corresponding variable-encoding or 
clause-encoding vertex is strictly smaller than \(N^2 + 1\).
Consequently, for any two vertices contained in the same part of \(V_c\), the shortest path 
between them passes through a portal vertex and is strictly shorter than any 
path routed via the semi-global vertex \(g\).
Therefore, the statement holds for any pair of vertices belonging to the same part of \(V_c\).
\end{claimproof}

In the final claim, we consider the remaining cases of possible parts in which vertices \(u\) and \(v\) are present.

\begin{claim}\label{claim:edges-in-H} 
Consider two vertices \(u\) and \(v\) that do not belong to the same part in \(V_c\),
and there exists a path from \(u\) to \(v\) in \(G - \{g\}\) that contains exactly 
one portal vertex. 
Then \(u\) and \(v\) are mutually maximally distant in \(G\) if and only if
they are adjacent in \(H\).
\end{claim}
\begin{claimproof}
We distinguish the possible cases according to the parts containing \(u\) and \(v\).

\begin{enumerate}[nolistsep]
\item\label{item:green-edges} (Green Edges) We first treat the case \(u \in X_T\) and \(v \in X_F\).
\end{enumerate}
Let \(u = x_{i,t}\) and \(v = x_{j,f}\) for some \(i,j \in [n]\).
Recall that the only neighbors of \(x_{j,f}\) are \(\g(x_{j,f})\) and \(\f(x_{j,f})\).
Hence, any shortest path from \(x_{i,t}\) to \(x_{j,f}\) must pass through one of these two vertices.
First, consider the path from \(x_{i,t}\) to \(\g(x_{j,f})\) via the semi-global vertex \(g\).
By construction, this path has length \(2N^2 + 1\).
Next, consider the paths from \(x_{i,t}\) to \(\f(x_{j,f})\) through the portal vertex \(\alpha^{X_T,X_F}\).
The length of this path equals
\(1 + (N^2 + iN) + (N^2 - jN) = 2N^2 + 1 + (i-j)N\).
Similarly, the path through the portal vertices \(\beta^{X_T,X_F}\) has length
\(1 + (N^2 - iN) + (N^2 + jN) = 2N^2 + 1 - (i-j)N\).
Consequently, the shortest path from \(x_{i,t}\) to \(\f(x_{j,f})\) via a portal has length
\(2N^2 + 1 - |i-j|N\).
In particular, this value equals \(2N^2 + 1\) if and only if \(i = j\).

We now compare the distances.
If \(i \neq j\), then
\(2N^2 + 1 - |i-j|N < 2N^2 + 1\),
and hence every shortest path from \(x_{i,t}\) to \(x_{j,f}\) goes through \(\f(x_{j,f})\).
In this case,
\(\dist_G(x_{i,t}, \f(x_{j,f})) < \dist_G(x_{i,t}, \g(x_{j,f}))\),
and therefore \(x_{j,f}\) is not maximally distant from \(x_{i,t}\).
On the other hand, if \(i = j\), then both routes (via \(g\) and via the portal)
have equal length \(2N^2 + 1\).
Hence,
\(\dist_G(x_{i,t}, \f(x_{j,f})) = \dist_G(x_{i,t}, \g(x_{j,f}))\),
and thus \(x_{j,f}\) is maximally distant from \(x_{i,t}\).
By symmetry of the construction, the same holds in the reverse direction,
and therefore \(x_{i,t}\) and \(x_{j,f}\) are mutually maximally distant if and only if \(i = j\).
By the definition of \(H\), we add an edge \(x_{i,t}x_{j,f}\) precisely when \(i = j\).
Therefore, in the case \(u \in X_T\) and \(v \in X_F\),
the vertices \(u\) and \(v\) are mutually maximally distant in \(G\) if and only if
they are adjacent in \(H\).

The cases \(u \in Y_T\) and \(v \in Y_F\), as well as
\(u \in Z_T\) and \(v \in Z_F\), follow by identical arguments.

\begin{enumerate}[resume,nolistsep]
\item\label{item:red-edges} (Red Edges) 
We consider the case when \(u \in X_T\) and \(v \in P_T\).
\end{enumerate}
The proof proceeds analogously to the case \(u \in X_T\) and \(v \in X_F\),
with the only difference being that the relevant portal vertices are
\(\alpha^{X_T,P_T}\) and \(\beta^{X_T,P_T}\).

Let \(u = x_{i,t}\) and \(v = x^{\lambda}_{j,t}\), where \(i, j \in [n]\) and \(\lambda \in [m]\).
As in the previous case, the only neighbors of \(x^{\lambda}_{j,t}\) are
\(\g(x^{\lambda}_{j,t})\) and \(\f(x^{\lambda}_{j,t})\).
Hence, every shortest path from \(x_{i,t}\) to \(x^{\lambda}_{j,t}\)
must pass through one of these two vertices.
A direct computation of the lengths of the paths via
\(\alpha^{X_T,P_T}\) and \(\beta^{X_T,P_T}\)
shows that the shortest portal path has length
\(2N^2 + 1 - |i - j|N\),
whereas the path via the semi-global vertex \(g\)
has length \(2N^2 + 1\).
We recall that the length of path connecting \(x^{\lambda}_{j,t}\)
to \(\alpha^{X_T,P_T}\) and \(\beta^{X_T,P_T}\) depends on \(j\) and not 
on \(\lambda\).
Consequently, these two values are equal if and only if \(i = j\).

Arguing exactly as before, we conclude that
\(x_{i,t}\) and \(x^{\lambda}_{j,t}\) are mutually maximally distant in \(G\)
if and only if \(i = j\).
By the construction of \(H\), this holds precisely when
\(x_{i,t}\) and \(x^{\lambda}_{j,t}\) are adjacent in \(H\).
The same argument extends verbatim to the remaining pairs of parts,
namely \((X_F,P_F)\), \((Y_T,Q_T)\), \((Y_F,Q_F)\),
\((Z_T,R_T)\), and \((Z_F,R_F)\).

\begin{enumerate}[resume,nolistsep]
\item\label{item:purple-edges} (Purple Edges) We consider the case when \(u \in P_T\) and \(v \in Q_T\). 
\end{enumerate}
The proof proceeds analogously to the previous cases 
with the only difference being that the relevant portal vertices are
\(\alpha^{P_T,Q_T}\) and \(\beta^{P_T,Q_T}\).

Let \(u = x^{\lambda}_{i,t}\) and \(v = y^{\lambda_1}_{j,t}\), where \(i, j \in [n]\) and \(\lambda, \lambda_1 \in [m]\).

A direct computation of the lengths of the paths via
\(\alpha^{P_T,Q_T}\) and \(\beta^{P_T,Q_T}\)
shows that the shortest portal path has length
\(2N^2 + 1 - |\lambda - \lambda_1|N\),
whereas the path via the semi-global vertex \(g\)
has length \(2N^2 + 1\).
We recall that the length of path connecting \(x^{\lambda}_{i,t}\)
to \(\alpha^{P_T,Q_T}\) and \(\beta^{P_T,Q_T}\) depends on \(\lambda\) and not 
on \(i\).
Consequently, these two values are equal if and only if \(\lambda = \lambda_1\).

Arguing exactly as before, we conclude that
\(x^{\lambda}_{i,t}\) and \(y^{\lambda}_{j,t}\) are mutually maximally distant in \(G\)
if and only if \(\lambda = \lambda_1\).
By the construction of \(H\), this holds precisely when
variables corresponding to \(x^{\lambda}_{i,t}\) and \(y^{\lambda_1}_{j,t}\) are in 
the same clause and hence are adjacent in \(H\).
The same argument extends to the remaining pairs of parts
which we group for better understanding:
\(\{(P_T, Q_F), (P_T, R_T), (P_T, R_F) \}\),
\(\{(P_F, Q_T), (P_F, Q_F), (P_F, R_T), (P_F, R_F) \}\),
\(\{(Q_T, R_T), (Q_T, R_F) \}\), and
\(\{(Q_F, R_T), (Q_F, R_F) \}\).

\begin{enumerate}[resume,nolistsep]
\item\label{item:no-edges} (No Edges) We consider the case when \(u \in P_T\) and \(v \in P_F\). 
\end{enumerate}
Let \(u = x^{\lambda}_{i,t}\) and \(v = x^{\lambda_1}_{j, f}\) where \(i, j \in [n]\) and
\( \lambda, \lambda_1 \in [m]\).
Note that the shortest path from \(u = x^{\lambda}_{i,t}\) to \(v = x^{\lambda_1}_{j, f}\)
via \(\beta^{P, Q}\) (or \(\alpha^{R, P}\)) is of length \(2N^2 + 2 - |\lambda + \lambda_1|N\).
See Table~\ref{table:lenghts-clause-to-clause}.
However, the path from  \(u = x^{\lambda}_{i,t}\) to \(v = x^{\lambda_1}_{j, f}\)
via semi-global vertex \(g\) is of length \(2N^2 + 2\).
Hence, for any \(\lambda, \lambda_1 \in [m]\), the path via portal vertex is smaller.
Hence, \(\dist_G(x^{\lambda}_{i,t}, \f(x^{\lambda_1}_{j,f}))  <
\dist_G(x^{\lambda}_{i,t}, x^{\lambda_1}_{j,f}) < \dist_G(x^{\lambda}_{i,t}, \g(x^{\lambda_1}_{j,f}))\).
Hence,  \(u = x^{\lambda}_{i,t}\) is not maximally distant from \(v = x^{\lambda_1}_{j,f}\).
Note that by the construction of \(H\), there is no edge with one endpoint in \(P_T\)
and another endpoint in \(P_Q\).

The similar argument follows when \(u \in Q_T\) and \(v \in Q_F\), 
and when \(u \in R_T\) and \(v \in R_F\).

This proves that any two vertices \(u\) and \(v\) that do not belong to the same part in \(V_c\),
and there exists a path from \(u\) to \(v\) in \(G - \{g\}\) that contains exactly 
one portal vertex are mutually maximally distant in \(G\) if and only if
they are adjacent in \(H\).
\end{claimproof}

The above three claims prove that vertices \(u, v\) in \(V(H) = V_c\) are adjacent in \(H\)
if and only if there are mutually maximal distant from each other, i.e., \(E(H) = E(G_{SR})\).
By Observation~\ref{obs:wlog-vertex-cover},
we can assume without loss of generality that
every vertex cover \(C \subseteq V(G_{SR})\) contains all vertices in \(V_p \setminus \{p\}\).
This implies that \(C \subseteq V(G_{SR})\) is a vertex cover of \(G_{SR}\) if and only if 
\(C \setminus V_p\) is a vertex cover of \(H\),
which concludes the proof of Lemma~\ref{lemma:vertex-cover-two-graphs}.

\subsection{Correctness of the reduction and Proof of Theorem~\ref{thm:np-hardness-fvs}}

We are now in the position to prove that the reduction mentioned in Section~\ref{subsec:3SAT-to-SMD}
is correct.
By Proposition~\ref{prop:smd-vc}, the instance \((G, k)\) is a yes-instance of 
\textsc{Strong Metric Dimension} if and only if \((G_{SR}, k)\) is a yes-instance 
of \textsc{Vertex Cover}. 
By Lemma~\ref{lemma:Ex-3-Par-3-SAT-to-VC}, 
\((\psi, \langle X, Y, Z\rangle)\) is a yes-instance of \textsc{Exact-3-Partitioned-3-SAT} 
if and only if \((H, 3n + 2m)\) is a yes-instance of \textsc{Vertex Cover}. 
By Lemma~\ref{lemma:vertex-cover-two-graphs}, \((G_{SR}, k)\) 
is a yes-instance of \textsc{Vertex Cover} if and only if \((H, 3n + 2m)\) is a yes-instance
of \textsc{Vertex Cover}.
Combining these equivalences, we obtain the following lemma.

\begin{lemma}\label{lemma:3SAT-to-SMD}
$(\psi, \langle X, Y, Z\rangle)$ is a yes-instance of \textsc{Exact-3-Partitioned-3-SAT}
if and only if $(G,k)$ is a yes-instance of \textsc{Strong Metric Dimension}.  
\end{lemma}

It is easy to verify that all the reductions  mentioned in this section take 
the time polynomial in the size of input.
Also, as specified in Section~\ref{subsec:3SAT-to-SMD},
the feedback vertex set number of \(G\) is \(25\) whereas its
pathwidth is at most \(27\).
This concludes the proof of Theorem~\ref{thm:np-hardness-fvs}.
\section{Conclusion}

Foucaud et al.~\cite{DBLP:conf/icalp/FoucaudGK0IST24} showed that 
\mdfull, \gsfull, and \smdfull\ are \NP-complete problems that admit
a double-exponential lower bound and a matching algorithm
when parameterized by treewidth plus diameter and by the vertex cover number
of the graph.
They justify the study of these parameters by highlighting the fact that
\mdfull\ and \gsfull\ remain \NP-complete even when restricted to graphs of
diameter two or to graphs of constant treewidth.
However, such results were not previously known for \smdfull.
In this article, we bridge this gap by proving that \smdfull\ is \NP-complete
even when restricted to graphs of diameter two, as well as when restricted to
graphs of constant pathwidth plus feedback vertex set number.

Our results for \smdfull\ are analogous to those obtained for \gsfull\
in~\cite{DBLP:conf/iwpec/Tale25}.
This raises the question whether a similar result holds for \mdfull.
More precisely, is \mdfull\ \NP-complete even on graphs of constant
pathwidth plus feedback vertex set number?
If true, this would strengthen the result of~\cite{GKIST23}, which shows that
the problem is \W[1]-hard when parameterized by this combined parameter.


\bibliography{bib}

@inproceedings{DBLP:conf/stacs/FoucaudGK0IST25,
  author       = {F. Foucaud and
                  E. Galby and
                  L. Khazaliya and
                  S. Li and
                  F. Mc Inerney and
                  R. Sharma and
                  P. Tale},
  title        = {Metric Dimension and Geodetic Set Parameterized by Vertex Cover},
  booktitle    = {42nd International Symposium on Theoretical Aspects of Computer Science,
                  {STACS}},
  series       = {LIPIcs},
  pages        = {33:1--33:20},
  year         = {2025}
}

@article{Courcelle90,
  author  = {B. Courcelle},
  title   = {The Monadic Second-Order Logic of Graphs. {I.} Recognizable Sets of
             Finite Graphs},
  journal = {Inf. Comput.},
  volume  = {85},
  number  = {1},
  pages   = {12--75},
  year    = {1990}
}

@book{D12,
  author    = {R. Diestel},
  title     = {Graph Theory, 4th Edition},
  series    = {Graduate texts in mathematics},
  volume    = {173},
  publisher = {Springer},
  year      = {2012}
}

@book{DBLP:books/daglib/0015106,
  author    = {J. M. Kleinberg and
               E. Tardos},
  title     = {Algorithm design},
  publisher = {Addison-Wesley},
  year      = {2006}
}

@book{DBLP:books/fm/GareyJ79,
  author    = {M. R. Garey and
               D. S. Johnson},
  title     = {Computers and Intractability: {A} Guide to the Theory of NP-Completeness},
  publisher = {W.H. Freeman and Company},
  year      = {1979},
  isbn      = {0-7167-1044-7},
  bibsource = {dblp computer science bibliography, https://dblp.org}
}

@book{DBLP:books/sp/CyganFKLMPPS15,
  author    = {M. Cygan and
               F. V. Fomin and
               L. Kowalik and
               D. Lokshtanov and
               D. Marx and
               M. Pilipczuk and
               M. Pilipczuk and
               S. Saurabh},
  title     = {Parameterized Algorithms},
  publisher = {Springer},
  year      = {2015}
}

@inproceedings{DBLP:conf/icalp/FoucaudGK0IST24,
  author    = {F. Foucaud and
               E. Galby and
               L. Khazaliya and
               S. Li and
               F. Mc Inerney and
               R. Sharma and
               P. Tale},
  title     = {Problems in {NP} Can Admit Double-Exponential Lower Bounds When Parameterized
               by Treewidth or Vertex Cover},
  booktitle = {51st International Colloquium on Automata, Languages, and Programming,
               {ICALP}},
  series    = {LIPIcs},
  volume    = {297},
  pages     = {66:1--66:19},
  year      = {2024},
  bibsource = {dblp computer science bibliography, https://dblp.org}
}

@inproceedings{DBLP:conf/iwpec/Tale25,
  author    = {P. Tale},
  title     = {Geodetic Set on Graphs of Constant Pathwidth and Feedback Vertex Set Number},
  booktitle = {20th International Symposium on Parameterized and Exact Computation, {IPEC}},
  series    = {LIPIcs},
  volume    = {358},
  pages     = {28:1--28:14},
  year      = {2025},
  bibsource = {dblp computer science bibliography, https://dblp.org}
}

@article{DBLP:journals/dam/DasGuptaM17,
  author  = {B. DasGupta and
             N. Mobasheri},
  title   = {On optimal approximability results for computing the strong metric
             dimension},
  journal = {Discrete Applied Math.},
  volume  = {221},
  pages   = {18--24},
  year    = {2017}
}

@article{DBLP:journals/dam/KuziakPRY18,
  author  = {D. Kuziak and
             M. L. Puertas and
             J. A. Rodr{\'{\i}}guez{-}Vel{\'{a}}zquez and
             I. G. Yero},
  title   = {Strong resolving graphs: The realization and the characterization
             problems},
  journal = {Discrete Applied Math.},
  volume  = {236},
  pages   = {270--287},
  year    = {2018}
}

@article{DBLP:journals/dam/OellermannP07,
  author  = {O. R. Oellermann and
             J. Peters{-}Fransen},
  title   = {The strong metric dimension of graphs and digraphs},
  journal = {Discrete Applied Math.},
  volume  = {155},
  number  = {3},
  pages   = {356--364},
  year    = {2007}
}

@article{DBLP:journals/siamdm/LampisMV25,
  author    = {Michael Lampis and
               Nikolaos Melissinos and
               Manolis Vasilakis},
  title     = {Parameterized Max Min Feedback Vertex Set},
  journal   = {{SIAM} J. Discret. Math.},
  volume    = {39},
  number    = {3},
  pages     = {1587--1620},
  year      = {2025},
  url       = {https://doi.org/10.1137/23m1605247},
  doi       = {10.1137/23M1605247},
  bibsource = {dblp computer science bibliography, https://dblp.org}
}

@inproceedings{floCALDAM20,
  author    = {D. Chakraborty and
               F. Foucaud and
               H. Gahlawat and
               S. K. Ghosh and
               B. Roy},
  title     = {Hardness and Approximation for the Geodetic Set Problem in Some Graph Classes},
  booktitle = {6th International Conference on Algorithms and Discrete Applied Mathematics ({CALDAM})},
  volume    = {12016},
  year      = {2020},
  pages     = {102--115}
}

@article{FoucaudMNPV17b,
  author  = {F. Foucaud and
             G. B. Mertzios and
             R. Naserasr and
             A. Parreau and
             P. Valicov},
  title   = {Identification, Location-Domination and Metric Dimension on Interval
             and Permutation Graphs. {II.} {A}lgorithms and Complexity},
  journal = {Algorithmica},
  volume  = {78},
  number  = {3},
  pages   = {914--944},
  year    = {2017}
}

@article{GKIST23,
  author  = {E. Galby and
             L. Khazaliya and
             F. {Mc~Inerney} and
             R. Sharma and
             P. Tale},
  title   = {Metric Dimension Parameterized by Feedback Vertex Set and Other Structural
             Parameters},
  journal = {SIAM J. Discrete Math.},
  volume  = {37},
  number  = {4},
  pages   = {2241--2264},
  year    = {2023}
}

@inproceedings{kuziak2014strong,
  title        = {Strong resolvability in product graphs},
  author       = {D. Kuziak},
  booktitle    = {1st URV Doctoral Workshop in Computer Science and Mathematics},
  pages        = {41},
  year         = {2014},
  organization = {Publicacions URV}
}

@article{LM21,
  author  = {S. Li and M. Pilipczuk},
  title   = {Hardness of Metric Dimension in Graphs of Constant Treewidth},
  journal = {Algorithmica},
  pages   = {3110--3155},
  year    = {2022},
  volume  = {84},
  number  = {11}
}

@article{sebo04,
  title   = {On Metric Generators of Graphs},
  author  = {A. Seb\H{o} and E. Tannier},
  journal = {Mathematics of Operations Research},
  volume  = {29},
  number  = {2},
  pages   = {383--393},
  year    = {2004}
}

\end{document}